\documentclass [12pt]{article}

\usepackage{amsmath,amsthm,amscd,amssymb}

\usepackage[small, centerlast, it]{caption}
\usepackage{epsfig}
\usepackage[titles]{tocloft}
\usepackage[latin1]{inputenc}
\usepackage{verbatim} 
\usepackage[active]{srcltx} 
\usepackage{fancyhdr}
\usepackage{caption}

\def\b1{{1\!\!1}}

\def\cB{{\ca B}}
\def\cC{{\ca C}}

\def\cE{{\ca E}}

\def\cL{{\ca L}}
\def\cM{{\ca M}}

\def\cS{{\ca S}}

\def\sH{{\mathsf H}}
\def\sP{{\mathsf P}}

\def\sS{{\mathsf S}}

\def\bC{{\mathbb C}}           

\def\bI{{\mathbb I}}

\def\bN{{\mathbb N}}

\def\bR{{\mathbb R}}

\def\gF{{\mathfrak F}}

\def\gP{{\mathfrak P}}

\def\beq{\begin{eqnarray}}
\def\eeq{\end{eqnarray}}

\newcommand{\ca}[1]{{\cal #1}}         

\def\tr{\mbox{tr}}

\def\p{\parallel}

\newcommand{\bra}[1]{\langle{#1}|}
\newcommand{\ket}[1]{|{#1}\rangle}

\usepackage{sectsty}
\sectionfont{\normalsize}
\subsectionfont{\normalsize\normalfont\itshape}
\newtheoremstyle{thm}
{12pt}
{12pt}
{\itshape}
{}
{\itshape\bfseries}
{}
{1em}
{}

\theoremstyle{thm}
\newtheorem{theorem}{Theorem}

\newtheorem{proposition}[theorem]{Proposition}
\newtheorem{definition}[theorem]{Definition}

\begin{document}

\par 
\bigskip 
\par 
\rm


\par
\bigskip
\large
\noindent
{\bf Geometric viewpoint on the quantization of a fuzzy logic}
\bigskip
\par
\rm
\normalsize 


\noindent   {\bf Davide Pastorello}\\
\par
\noindent Department of Information Engineering and Computer Science, University of Trento \\via Sommarive 19, 38123 Povo (Trento), Italy\\
e-mail: d.pastorello@unitn.it\\
 \normalsize

\par

\rm\normalsize


\rm\normalsize


\par
\bigskip

\noindent
\small
{\bf Abstract}. Within the Hamiltonian framework, the propositions about a classical physical system are described in the Borel $\sigma$-algebra of a symplectic manifold (the phase space) where logical connectives are the standard set operations. 
Considering the geometric formulation of quantum mechanics we give a description of quantum propositions in terms of fuzzy events in a complex projective space equipped with K\"ahler structure (the quantum phase space) obtaining a quantized version of a fuzzy logic by deformation of the product $t$-norm.

\normalsize

\section{Introduction}

In quantum experiments measurement processes alter the observed objects and simultaneous measurements of two physical quantities are sometimes impossible. These phenomenological evidences are sufficient to preclude the use of Boolean logic to describe propositions about quantum systems. In \cite{BvN} Birkhoff and von Neumann formulated the \emph{standard quantum logic} describing quantum propositions in the non-distributive lattice of orthogonal projectors on a Hilbert space where any sublattice of commuting projectors presents the structure of a Boolean $\sigma$-algebra. In the following decades some celebrated milestones in quantum logic were the Mackey's program \cite{Mac}, the Piron's axiomatization \cite{piron}, the work of Foulis and Randall on empirical logic \cite{FR}. In more recent years the field of quantum logic has been related to a variety of abstract structures that generalize the archetypical lattice of projectors like orthomodular posets, orthoalgebras, effect algebras and categories \cite{HQL}. There are also several geometric approaches to quantum logic like \cite{convex} or the recent works \cite{CK} and \cite{camosso}.     

In this work we adopt the geometric viewpoint on quantum mechanics, in the sense of \cite{Kibble,AS,BH,BSS,DV2} whereby quantum systems can be geometrically described in a classical-like fashion, to introduce a non-distributive orthocomplemented structure to represent quantum propositions in analogy with classical mechanics. Propositions about a classical system, represented by Borel sets in the phase space, form a Boolean $\sigma$-algebra where logical conjunction is the intersection, logical disjunction is the union, logical implication is the inclusion and the negation is the set complement. In the quantum case we observe how propositions can be represented by a special class of fuzzy events (fuzzy sets with Borel measurable membership functions on which a probability measure can be defined) in the quantum phase space. In geometric quantum mechanics the quantum phase space is given by a projective Hilbert space, so we are within the \emph{Hilbertian scenario} provided by the Birkhoff-von Neumann quantum logic. In this sense we obviously construct a quantum logic that is isomorphic to the lattice of orthogonal projectors on the underlying Hilbert space. However the obtained quantum logic turns out to be a special case of an orthomodular poset made by fuzzy sets characterized by a non-commutative product defined over the membership functions. So the description of quantum propositions suggested by the strong analogy between Hamiltonian classical mechanics and geometric quantum mechanics leads to a result that is connected to general quantum logics formed by fuzzy sets (objects introduced in \cite{py1}) beyond the notion of projective Hilbert space and the geometric formulation itself.

In Section \ref{props} we give a very short overview on propositions about a classical system as Borel sets in the phase space and on the need to construct non-Boolean structures to describe quantum propositions. Section \ref{GQM} is devoted to geometric quantum mechanics and its formal analogy to Hamiltonian classical mechanics. Section \ref{FSE} introduces some ideas about fuzzy sets that are relevant to represent quantum propositions in the quantum phase space. In section \ref{FECPS} we develop a concrete quantum logic by analogy with the classical case in the geometric framework, the result is a particular class of fuzzy events with idempotent membership functions w.r.t.\!\! a non-commutative product. In Section \ref{deform t-norm} we generalize the quantum logic of fuzzy events in the complex projective space introducing a general quantum structure that can be obtained from a collection of fuzzy sets defining non-standard set operations, in particular replacing the pointwise product in the $t$-norm with a non-commutative product. Finally we compare the obtain general result with some aspects of deformation quantization.

\section{Propositions about physical systems: Classical and Quantum}\label{props}

Informally, a \emph{proposition about a physical system} is a statement that can be verified by a measurement process on the system. For example, if the considered system is a particle moving on the real line a proposition is: "\emph{The position of the particle falls in $[0,1]$}". If the system is classical, we can decide whether any proposition is true or false with certainty if we have the complete knowledge of the system physical conditions (i.e.\!\!  we know its state). In the case we have a partial knowledge of the system conditions, we can just assign to each proposition a probability to be true. If the considered system is quantum, in general we cannot decide whether any proposition is true or false with certainty even if we have the maximum knowledge of its physical conditions (i.e.\! it is in a pure state). In this section we give a very short overview on the mathematical structures to describe propositions about classical and quantum systems following an outline similar to that adopted in \cite{V}.

Within the \emph{Hamiltonian formulation}, a classical system with $n$ spatial degrees of freedom is described in a $2n$-dimensional symplectic manifold $(\cM,\omega)$ called \emph{phase space}. A point in $\cM$ represents the complete knowledge of the system physical conditions. The time evolution of the system is represented by an integral curve in $\cM$ satisfying the \emph{Hamilton equation}:
\beq
\frac{dx}{dt}=X_H(x(t)),
\eeq
where $X_H$ is the \emph{Hamiltonian vector field} given by the \emph{Hamiltonian function} $H:\cM\rightarrow\bR$, that is the unique vector field such that $\omega(X_H, Y)=dH(Y)$ for any vector field $Y$ on $\cM$. The smooth function $H$ represents the total energy of the considered classical system. 

If there is a lack of knowledge about system conditions 
 then the state does not coincide with a single point of $\cM$ but it is represented by a smooth function $\rho:\cM\rightarrow[0,1]$ so that $\rho(x)$ is the probability that the exact physical condition of the system is represented by $x\in\cM$. According to this meaning, $\rho$ satisfies the normalization requirement:
\beq
\int_\cM \rho \,d\mu=1,
\eeq
where the measure $\mu$ on the Borel $\sigma$-algebra $\cB(\cM)$ coincides to the Lebesgue measure on every local chart of $\cM$.
In this case the dynamics is described by the \emph{Liouville equation}:
\beq
\frac{\partial \rho}{\partial t}+\{\rho, H\}_{PB}=0,
\eeq     
where the \emph{Poisson bracket} of a pair of real smooth functions $f$ and $g$ is defined as $\{f,g\}_{PB}:=\omega(X_f, X_g)$. 

From a general viewpoint a state of a classical system can be defined as a map that assigns to each proposition the probability to be true. Any proposition $P$ about the considered classical system can be represented by the set of points in $\cM$ that render $P$ true, in this sense a natural notion of classical state is a measure. On one hand a probabilistic state, described by a probability distribution $\rho$ can be defined by a Borel probability measure $\sigma:\cB(\cM)\rightarrow [0,1]$ in the following way:
\beq\label{state}
\sigma(A):= \int_A \rho \, d\mu\qquad A\in\cB(\cM).
\eeq
On the other hand a \emph{sharp} state, i.e. given by a single point in $x\in\cM$, can be defined as Dirac mass $\sigma:=\delta_x$ concentrated at $x\in\cM$. In this picture classical propositions can be identified as Borel sets in $\cM$ where logical connectives correspond to set operations. Given $P,Q\in\cB(\cM)$, $P\cup Q$  represents the disjunction of the corresponding propositions, $P\cap Q$ represents the conjunction and $\cM\setminus P$ represents the negation of proposition $P$. The inclusion $P\subset Q$ represents the implication $P\Rightarrow Q$. We can assume that any set of $\cB(\cM)$ represents a proposition about the considered classical system, otherwise we can assume that the propositions form a $\sigma$-algebra that is not the full $\cB(\cM)$. Anyway the propositions about a classical system present a structure of \emph{Boolean $\sigma$-algebra} w.r.t. the partial order relation given by set inclusion $\subset$. Let us introduce some basic notions about partially ordered sets (posets) and lattices in the following definition:

\begin{definition}
A poset $(L,\geq)$ is said to be:
\\
i) \textbf{bounded} if $L$ contains a minimum 0 and a maximum $\bI$; \\
ii) \textbf{orthocomplemented} if $L$ is equipped with a map $\neg:L\rightarrow L$, called \textbf{orthocomplementation}, satisfying:

a) $\neg (\neg p)=p$ for any $p\in L$,
 
b) $p\geq q\Rightarrow \neg q\geq \neg p$ for any $p,q\in L$,

c) the greatest lower bound $p\wedge \neg p$ and the least upper bound $p\vee \neg p$ exist in $L$ and

 $p\wedge \neg p=0$, $p\vee \neg p=\bI$;
\\
iii) \textbf{$\sigma$-orthocomplete} if every countable set $\{p_i\}_{i\in\bN}$ 
made by \textbf{orthogonal} elements, i.e. $\neg p_i\geq p_j$ (written $p_i\perp p_j$) for $i\not =j$, admits least upper bound $\vee_{i\in\bN} p_i\in L$.
\\
iv) \textbf{orthomodular} if $L$ is orthocomplemented and $q\geq p$ implies $q=p\vee (\neg p\wedge q)$.
\\
Two elements $p,q\in L$ are called \textbf{compatible} if $p=r_1\vee r_3$ and $q=r_2\vee r_3$ with $r_i\perp r_j$ for $i\not = j$.
\\
A poset $(L,\geq)$ is a \textbf{lattice} if for any $p,q\in L$ the greatest lower bound $p\wedge q$ and the least upper bound $p\vee q$ exist. 
\\
A lattice $(L,\geq)$ is said to be \textbf{distributive} if  
\beq\label{distrib}
p\vee(q\wedge r)=(p\vee q)\wedge (p\vee r) \qquad\mbox{and} \qquad p\wedge(q\vee r)=(p\wedge q)\vee (p\wedge r) 
\eeq
 for any $p,q,r\in L$.
\\
A \textbf{Boolean algebra} is a lattice that is distributive, bounded, orthocomplemented (hence orthomodular). A \textbf{Boolean $\sigma$-algebra} is a Boolean algebra such that any countable subset admits least upper bound.

 \end{definition}

Let us recall some remarkable facts: 1) In any orthocomplemented lattice \emph{De Morgan laws} are satisfied. 2) Any $\sigma$-algebra $X$ of sets is a Boolean $\sigma$-algebra where the partial order relation is the set inclusion, so $\vee$ and $\wedge$ correspond to $\cup$ and $\cap$ respectively, the maximum is $X$ and the minimum is $\emptyset$, the orthocomplementation of $A\in X$ is $\neg A=X \setminus A$. 3) Any pair of elements in a Boolean algebra is compatible. 4) A Boolean algebra can be equivalently defined as a six-tuple $(L, \wedge, \vee, \neg, 0, \bI)$ satisfying the following requirement:
\\
i) $\wedge$ and $\vee$ are commutative associative binary operations on $L$;
\\
ii) $\wedge$ and $\vee$ are distributive (\ref{distrib});
\\
iii) $p=p\vee(p\wedge q)=p\wedge(p\vee q)$ for any $p,q\in L$;
\\
iv) $p\vee 0=p$ and $p\wedge \bI=p$ for any $p\in L$;
\\
v) $\neg$ is a unary operation on $L$ such that $p\vee\neg p=1$ and $p\wedge \neg p=0$ for any $p\in L$.

In the case of quantum systems, if we need a model to describe propositions  we must consider some phenomenological evidences like quantum randomness (and relative non-epistemic uncertainty) and the existence of physical quantities that cannot be measured simultaneously (like position and momentum of a quantum particle, or linear and circular polarizations of a photon), so there are pairs of propositions that cannot be simultaneously verifiable and meaningfully combined via logical connectives. Therefore the structure of Boolean $\sigma$-algebra turns out to be too severe for quantum propositions. Let us relax the structure of Boolean $\sigma$-algebra and consider a bounded $\sigma$-orthocomplete orthomodular poset where a notion of \emph{generalized probability measure}, i.e.\!\! a notion of state, can be defined requiring normalization and additivity properties. We adopt the general definition of \emph{quantum logic} given in \cite{py3}. 
\begin{definition}\label{GPM}
A \textbf{generalized probability measure} on a bounded $\sigma$-orthocomplete orthomodular poset $L$ is a map $\sigma:L\rightarrow[0,1]$ satisfying:
\\
i) $\sigma(\bI)=1$;\\
ii) For any countable set $\{p_i\}_{i\in\bN}$ such that $p_i\perp p_j$ for $i\not =j$:
$$\sigma(\vee_{i\in\bN} p_i)=\sum_{i\in\bN}\sigma(p_i).$$
Let $\cS$ be a set of generalized probability measures on $L$, $\cS$ is called \textbf{ordering set} if $\sigma(p)\geq \sigma(q)$ for any $\sigma\in\cS$ implies $p\geq q$. 
\\
A \textbf{quantum logic} is a bounded $\sigma$-orthocomplete orthomodular poset $L$ equipped with an ordering set of generalized probability measures on $L$. 
\end{definition}

\noindent 
The elements of a quantum logic represent propositions about a quantum system and the compatible elements represent the propositions that can be simultaneously verifiable. The partial order relation between compatible propositions represents the logical implication, from an operational viewpoint in order to establish that a proposition implies another proposition one must test them varying the state of the observed quantum system, thus only a poset allowing for an ordering set of states can represent quantum propositions with physical meaning.   
The most known example of quantum logic is given by the non-distributive lattice $\gP(\sH)$ of orthogonal projectors in a separable Hilbert space $\sH$, that we call \emph{standard B-vN quantum logic}. The partial order relation on $\gP(\sH)$ is the inclusion of projection subspaces int $\sH$, the maximum $\bI$ is the identity operator, the minimum is the null operator and the orthocomplement of $P\in\gP(\sH)$ is $\neg P=\bI-P$.
In $\gP(\sH)$ the compatible elements are the orthogonal projectors that commute. The generalized probability measures on $\gP(\sH)$ are in bijective correspondence with the density matrices\footnote{A density matrix $\rho$ on $\sH$ is a positive trace class operator such that $\tr(\rho)=1$.} on $\sH$, as provided by Gleason's theorem \cite{Gleason}.

The main part of the present paper is devoted to the construction of a concrete quantum logic within the geometric formulation of quantum mechanics in analogy to the classical case sketched in this section. We have a quantum phase space given by a complex projective space, in order to describe quantum propositions as subsets of the phase space we must consider fuzzy sets in view of the randomness of outcomes in quantum experiments.

\section{Geometric quantum mechanics}\label{GQM}

This section is devoted to summarize some basics of geometric quantum mechanics. We focus on the finite-dimensional case where the geometric formulation is fully well-posed, however this is not a severe drawback in view of the generalization that we propose further where the notion of Hilbert space is abandoned.

Let $\sH$ be a $n$-dimensional Hilbert space with $n>2$. The projective Hilbert space over $\sH$ is defined as $\sP(\sH):=\frac{\sH}{\sim}\setminus [0]$ where, for $\psi,\varphi\in\sH$, $\psi\sim\varphi$ if and only if $\psi=\alpha\varphi$ with $\alpha\in \bC\setminus\{0\}$. $\sP(\sH)$ is connected and Hausdorff in the quotient topology. It is well-known that the map $\sP(\sH)\ni [\psi]\mapsto \ket\psi\bra\psi\in \mathfrak P_1(\sH)$, with $\p\psi\p=1$, is a homeomorphism where $\mathfrak P_1(\sH)$ is the space of rank-1 orthogonal projectors in $\sH$ equipped with the topology induced by the standard operator norm.  

The projective Hilbert space $\sP(\sH)$, that can be identified with the complex projective space $\bC P^{n-1}$, has a structure of a $(2n-2)$-dimensional smooth real manifold and the tangent vectors $v\in T_p\sP(\sH)$ have the form $v=-i[A_v, p]$ for some self-adjoint operator $A_v$ on $\sH$ \cite{DV2}. As a real manifold $\sP(\sH)$ can be equipped with a symplectic structure given by the following form:
\beq\label{omega}
\omega_p(u,v):=-i\,\tr([A_u, A_v]p).
\eeq
$\sP(\sH)$ can be also equipped with the Riemannian structure induced by the well-known Fubini-Study metric $g$:
\beq\label{g}
g_p(u,v):=-\tr(([A_u,p][A_v,p]+[A_v,p][A_u,p])p).
\eeq 
One can prove that the metric $g$ is compatible with the symplectic form $\omega$ by means of the complex form $j_p:T_p\sP(\sH)\ni v\mapsto i[v,p]\in T_p\sP(\sH)$, i.e. $\sP(\sH)$ is a K\"ahler manifold.\\
{$\sP(\sH)$ carries the following representation of the unitary group $U(n)$:
\beq\label{UonP}
U(n)\times\sP(\sH)\ni (U,p)\mapsto UpU^{-1}\in\sP(\sH).
\eeq} 
As proved in \cite{DV2}, the unique regular Borel measure $\nu$ that is left-invariant w.r.t. the smooth action (\ref{UonP}) of the unitary group $U(n)$ on $\sP(\sH)$, with $\nu(\sP(\sH))=1$, coincides to the Riemannian measure induced by the metric $g$ and to the Liouville volume form defined by $\omega \wedge \cdots\mbox{($n-1$) times}\cdots\wedge \omega$, where $\wedge$ is the wedge product\footnote{In this paper the symbol $\wedge$ denotes the logical conjunction or the greatest lower bound between two elements of a poset, except here.}, up to its normalization. 

The geometry of the complex projective space $\sP(\sH)$ allows to define a Poisson structure for the formulation of a Hamiltonian mechanics for quantum systems. In this sense $\sP(\sH)$ can be thought as a quantum phase space, in particular a single projective ray, i.e. a \emph{pure state}, represent the exact knowledge of the physical condition of a quantum system in analogy to a single point of the phase space of a classical system.

 Let $A$ be a self-adjoint operator on $\sH$, i.e.\! a quantum observable in standard quantum mechanics, we can describe it as a real function $f_A:\sP(\sH)\rightarrow\bR$ by the definition $f_A(p):=\tr(Ap)$. We have the equivalence between the Schr\"odinger dynamics induced by a self-adjoint operator operator $H$ and the Hamilton dynamics induced by the corresponding Hamiltonian function $f_H$ \cite{AS, BSS, DV2}, more precisely a curve in $\sP(\sH)$ is a solution of the Schr\"odinger equation if and only if it satisfies the Hamilton equation:
\beq
i\frac{dp}{dt}=[H,p] \quad \Leftrightarrow \quad \frac{dp}{dt}=X_{f_H}(p),
\eeq
where $X_{f_H}$ is the \emph{Hamiltonian vector field} of $f_H$ defined as the unique vector field on $\sP(\sH)$ such that $\omega(X_{f_H}, Y)=df_H$ for any vector filed $Y$ on $\sP(\sH)$. The quantum observables as phase space functions are completely characterized by the following statement \cite{AS}:
\begin{theorem}\label{AS}
Let $f:\sP(\sH)\rightarrow \bR$ be a smooth function. There is a self-adjoint operator $A$ on $\sH$ such that $f(p)=$\emph{tr}$(Ap)$ $\forall p\in\sP(\sH)$ if and only if the Hamiltonian vector field $X_f$ is a Killing vector field w.r.t. the Fubini-Study metric, that is $\mathcal L_{X_f}g=0$ where $\mathcal L$ is the Lie derivative.
\end{theorem}

The Poisson bracket of a pair of smooth functions $f,h:\sP(\sH)\rightarrow\bR$ is defined by $\{f,h\}_{PB}:=\omega(X_f, X_g)$. If $f(p):=\tr(A p)$ and $h(p):=\tr(B p)$ for two self-adjoint operators $A$ and $B$ on $\sH$ then the remarkable formula $\{f,h\}_{PB}=-i\tr([A,B]p)$ holds.

In addition to quantum observables also quantum states can be described by functions on $\sP(\sH)$ obtaining a classical-like framework to describe a quantum system. As not all the smooth real functions on $\sP(\sH)$ represent quantum observables we have that quantum states cannot be directly represented by probability densities on the quantum phase space \cite{DV2}. Let us call \emph{Liouville density} the function on $\sP(\sH)$ representing a quantum state as the analogue of a density probability of a probabilistic classical state. Given a density matrix $\sigma$ in $\sH$, i.e.\! a positive operator on $\sH$ with $\tr(\sigma)=1$, the corresponding {Liouville density} $\rho_\sigma:\sP(\sH)\rightarrow\bR$ is the unique Borel function on $\sP(\sH)$ such that:
\beq\label{densities}
\int_{\sP(\sH)} \rho_\sigma d\nu=1\quad\mbox{and}\quad
\tr(\sigma A)=\int_{\sP(\sH)} f_A\rho_\sigma \,d\nu,
\eeq    
for any self-adjoint operator $A$ on $\sH$, where $\nu$ is the $U(n)$-invariant Borel measure introduced above. The correspondence $\sigma\mapsto \rho_\sigma$, explicitly constructed in \cite{DV2}, is given by $\rho_\sigma(p)=n(n+1)\tr(\sigma p)-1$. Thus a Liouville density satisfying (\ref{densities}) cannot be interpreted directly as a probability density on $\sP(\sH)$ because it fails the requirement $0\leq \rho \leq 1$. Within this formulation the time evolution of mixed states is governed by the Liouville equation with Poisson bracket induced by (\ref{omega}).  It is clear that the Liouville density $\rho_{p_0}$ describing a pure state $p_0\in\sP(\sH)$ is not a Dirac delta centered in $p_0$, like a sharp classical state, but it is a \emph{smeared} distribution encoding the statistic produced by any possible measurement process on the system.

\section{Fuzzy sets and fuzzy events}\label{FSE}

Let us recall some basic notions of fuzzy set theory, historically proposed in \cite{Zadeh}: Let $\mathcal U$ be a non-empty set called \emph{universe of discourse} and $\mu_A:\mathcal U\rightarrow [0,1]$ be a function called \emph{membership function}. The pair $A=(\mathcal U, \mu_A)$ is called \emph{fuzzy set} and the value $\mu_A(x)$ is the \emph{grade of membership} of $x\in\mathcal U$ to $A$. We say that: $x$ is \emph{not included} in $X$ if $\mu_A(x)=0$,  $x$ is \emph{included} in $A$ if $\mu_A(x)=1$ and $x$ is \emph{partially included} in $A$ if $\mu_A(x)\in(0,1)$. If $x$ is partially included in $A$ it is also called \emph{fuzzy member} of $A$. Given two fuzzy sets $A$ and $B$, we say that $A$ is included in $B$ if $\mu_A(x)\leq \mu_B(x)$ for any $x\in\mathcal U$, the \emph{fuzzy inclusion} is denoted by $A\subseteq B$. There exist several ways for generalizing the classical set operations to fuzzy sets \cite{als}, let us recall the most known ones: A \emph{complement} of the fuzzy set $A$ is a fuzzy set $\neg A$ such that $\mu_{\neg A}(x)=1-\mu_A(x)$ for any $x\in\mathcal U$; \emph{union} and \emph{intersection} of the fuzzy sets $A$ and $B$ in the same universe $\mathcal U$ are fuzzy sets with membership functions respectively given by:
\beq\label{minnorm}
\mu_{A\cup B}(x)=\min\{\mu_A(x)+\mu_B(x), 1\}\quad , \quad \mu_{A\cap B}(x)=\max\{\mu_A(x)+\mu_B(x)-1, 0\}.
\eeq
Another generalization of classical set-theoretic union and intersection to fuzzy sets are given by the following membership functions:
\beq\label{product}
\mu_{A\cup B}(x)=\mu_A(x)+\mu_B(x)-\mu_A(x)\mu_B(x)\quad , \quad \mu_{A\cap B}(x)=\mu_A(x)\mu_B(x).
\eeq
In general, complements, unions, intersections of fuzzy sets can be constructed out following an axiomatic approach. Let us briefly recall the definition of fuzzy intersection in terms of a $t$-norm. A \emph{t-norm} is a function $t:[0,1]\times[0,1]\rightarrow[0,1]$ satisfying the following properties for all $x,y,w,z\in[0,1]$:
\\
i) $t(x,y)=t(y,x)$;\\
ii) $t(x,y)\leq t(z,w)$ if $x\leq z$ and $y\leq w$;\\
iii) $t(t(x,y),z)=t(z,t(y,z))$;\\
iv) $t(x,1)=1$.
\\
Given two fuzzy sets $A$ and $B$, their intersection $A\cap B$ is defined by:
\beq
\mu_{A\cap B}(x):=t(\mu_A(x), \mu_B(x)) \quad \forall x\in\mathcal U.
\eeq
The union $A\cup B$ is defined by means of the \emph{t-conorm} $s(x,y):=1-t(1-x,1-y)$ (this definition provides a generalization of De Morgan's laws):
\beq
\mu_{A\cup B}(x):=s(\mu_A(x), \mu_B(x)) \quad \forall x\in\mathcal U.
\eeq
The set operations (\ref{minnorm}) are defined by the so-called \emph{Lukasiewicz t-norm} $t(x,y):=\max\{x+y-1,0\}$ and the set operations (\ref{product}) are defined by the \emph{product t-norm} $t(x,y):=xy$. 

In \cite{Zadeh2} probability measures on fuzzy events are introduced: Let $\mathcal B(X)$ be the Borel $\sigma$-algebra on the topological space $X$ and $m:\mathcal B(X)\rightarrow [0,1]$ be a probability measure over $X$. A \emph{fuzzy event} in $X$ is a fuzzy set $A$ in $X$ whose membership function $\mu_A:X\rightarrow [0,1]$ is Borel measurable. The \emph{probability} of a fuzzy event $A$ is defined by:
\beq\label{pfe}
\mathbb P(A):=\int_X \mu_A(x) \, dm(x).
\eeq 
The integral (\ref{pfe}) is well-defined as $\mu_A$ is a Borel function, then we have a notion of probability measure on fuzzy events.

\section{Quantum propositions as fuzzy events in the complex projective space}\label{FECPS}

By analogy with the classical case, within the geometric formulation of quantum mechanics $\sP(\sH)$ plays the role of a \emph{quantum phase space}. Then we may try to develop an alternative  formulation of standard quantum logic based on measurable sets of $\sP(\sH)$ inspired by the classical case. Let us recall that in classical mechanics the propositions about a classical system described in the phase space $\cM$ are Borel sets in $\cM$ where logical connectives are represented by set operations. In this sense classical propositions present a structure of Boolean $\sigma$-algebra. 
As in classical mechanics, we have a phase space that is a symplectic manifold where we assume that any point represents a \emph{complete description} of the physical conditions of a system at given time. However in the quantum case we must to take into account of the following phenomenological evidences:
\\
\\
\textbf{E1. (Quantum randomness)} Repeated measurements of the physical quantity $A$ in the same physical condition produce different results.
\\
\\
\textbf{E2. (Incompatible observables)} There exist pairs of physical quantities that cannot be simultaneously measured.
\\

Let $P$ be a proposition about the considered quantum system, e.g. $P=$\emph{''The value of the physical quantity $A$ is $a\in\bR$"}, if we try to represent $P$ by a set $P\subset\sP(\sH)$ such that $p\in P$ if and only if $p$ renders $P$ true then we find an inconsistency with the evidence E1. In classical mechanics, given a proposition we can decide, for any point of the phase space, whether it is true; while in quantum mechanics, for any point of the quantum phase space, we just provide a probability that $P$ is true. This motivate the choice of representing quantum propositions by means of fuzzy sets on quantum phase space.
A quantum proposition $P$ can be represented by a fuzzy set in the universe $\sP(\sH)$ with the membership function $\mu_P:\sP(\sH)\rightarrow [0,1]$ so that the grade $\mu_P(p)$ is the probability that $p\in\sP(\sH)$ renders $P$ true. By analogy with the classical case, we define a state as a map assigning to any proposition $P$ the probability that it is true:
\beq\label{fuzzy state}
\sigma(P):=\int_{\sP(\sH)} \mu_P(p)\rho(p) \,d\nu(p),
\eeq 
where $\mu_P$ is the membership function of $P$, $\rho$ is a Liouville density on $\sP(\sH)$ and $\nu$ is the $U(n)$-invariant Borel measure defined in Section \ref{GQM}. In other words we define (\ref{fuzzy state}) substituting the indicator function of a Borel set in (\ref{state}) with the membership function of a fuzzy set. We can give this interpretation: In (\ref{fuzzy state}), $\mu_P$ represents the \emph{quantum uncertainty} and $\rho$ represents the \emph{epistemic uncertainty}. For $\sigma$ to be well-defined we require that $\mu_P$ is a Borel function on $\sP(\sH)$, in this sense quantum propositions are identified to fuzzy events in $\sP(\sH)$ and (\ref{fuzzy state}) is a probability of fuzzy events as in definition \ref{pfe} with $dm=\rho d\nu$. Moreover, by comparison with classical sharp states, we require that any Liouville density $\rho_{p_0}$ of a pure state, i.e. described by a single point $p_0\in\sP(\sH)$, acts as a Dirac delta on the membership functions of quantum propositions:
\beq\label{Fuzzy Dirac}
\sigma_{p_0}(P)=\int_{\sP(\sH)} \mu_P(p)\rho_{p_0}(p) \,d\nu(p)=\mu_P(p_0),
\eeq
obtaining a fuzzy version $\sigma_p$ of a Dirac measure (it provides a grade of membership instead of the value of an indicator function).
\begin{proposition} \label{mu}
A Borel measurable function $\mu:\sP(\sH)\rightarrow [0,1]$ satisfies
\beq \label{FD}
\int_{\sP(\sH)} \mu(p)\rho_{p_0}(p) \,d\nu(p)=\mu(p_0)\quad \forall p_0\in\sP(\sH)
\eeq
if and only if the Hamiltonian vector field $X_\mu$ of $\mu$ is defined and it is a Killing vector field w.r.t. Fubini-Study metric.
\end{proposition}
\begin{proof} $\{\psi_i\}_i$ be an orthonormal basis of the Hilbert space $\sH$ and be $\{p_i\}_i$ the family of corresponding projective rays $p_i=[\psi_i]$. As mentioned in Section \ref{GQM}, the form of a Liouville density $\rho_{p_i}$ satisfying the requirements (\ref{densities}) is $\rho_{p_i}(p)=n(n+1)\tr(p_i p)-1$ then $\sum_i \rho_{p_i}=n^2$. Assuming that (\ref{FD}) holds, we have:
$$\sum_i \mu_P(p_i)=\sum_i\int\mu(p)\rho_{p_i}(p)\, d\nu(p)=n^2\int_{\sP(\sH)} \mu(p) \,d\nu(p),
$$
then 
\beq\label{frame}
\sum_i \mu\left([\psi_i]\right)=\mbox{const}\quad\mbox{for any orthonormal basis $\{\psi_i\}_i$ of $\sH$}.
\eeq
Theorem 18 of \cite{DV2} reads that if $\mu\in\cL^2(\sP(\sH),\nu)$ and the property (\ref{frame}) is satisfied then there exists a unique self-adjoint operator $T$ on $\sH$ such that $\mu(p)=\tr(Tp)$ $\forall p\in\sP(\sH)$. Since $\mu\in\cL^1(\sP(\sH),\nu)\subset\cL^2(\sP(\sH),\nu)$ (because $\nu$ is finite and $\sP(\sH)$) is compact) we have $\mu(p)=\tr(Tp)$ $\forall p\in\sP(\sH)$ for a self-adjoint operator $T$ on $\sH$. The Hamiltonian vector field $X_\mu$ is a g-Killing vector field by Theorem \ref{AS}. Conversely, if $\mu$ is a smooth function with Hamiltonian vector field $X_\mu$ that is Killing, then $\mu(p)=\tr(Tp)$ for a self-adjoint operator $T$ in $\sH$ by Theorem \ref{AS}. Thus $\int_{\sP(\sH)} \mu\,\rho_{p_0}\, d\nu=\tr(Tp_0)=\mu(p_0)$ by formula (\ref{densities}). 
\end{proof}

In view of Proposition \ref{mu}, if we want to describe quantum randomness via fuzzy sets on the quantum phase space and quantum states as probability measures on fuzzy events, we assume that quantum propositions are represented by fuzzy events in $\sP(\sH)$ whose membership functions are smooth and present Hamiltonian vector fields that are Killing vector fields w.r.t. Fubini-Study metric. Now we need a structure that takes into account the existence of incompatible propositions related to the phenomenological evidence E2.
We need the structure of a quantum logic on the fuzzy events where compatible elements represents propositions that are simultaneously verifiable. The next result provides a particular procedure to obtain a quantum logic structure on the fuzzy events in $\sP(\sH)$, in the next section there is a more general result to endow a collection of fuzzy sets with a quantum logic structure. 

\begin{proposition}\label{geologic}
Let $\widetilde{\mathfrak F}$ be the class of fuzzy events in $\sP(\sH)$ with smooth membership functions whose Hamiltonian vector fields are Killing vector fields on $\sP(\sH)$. Let $\star$ be the non-commutative product defined on the real smooth functions on $\sP(\sH)$ by:
\beq f \star g := fg+\frac{i}{2}\{f, g\}_{PB} + \frac{1}{2} g(X_f,X_g) \qquad f,g\in\cC^\infty_\bR(\sP(\sH)), 
\eeq
where $\{\,\,,\,\,\}_{PB}$ is the Poisson bracket induced by the symplectic form on $\sP(\sH)$ and $g$ is the Fubini-Study metric.
The subclass $\gF$ of idempotent elements in $\widetilde\gF$ ($P\in\widetilde\gF$ is said to be idempotent if $\mu_P\star \mu_P=\mu_P$) is a quantum logic where: \\
i) the partial order relation is the fuzzy inclusion;\\
ii) the maximum and the minimum are given by $\sP(\sH)$ and $\emptyset$;\\
iii) the orthocomplementation is the fuzzy complement;\\
iv) two fuzzy events $P$ and $Q$ are compatible if and only if $\mu_P\star\mu_Q=\mu_Q\star \mu_P$;\\ 
v) $P$ and $Q$ are orthogonal if and only if $\mu_P\star\mu_Q=0$;\\
vi) If $P$ and $Q$ are compatible then $P\vee Q$ and $P\wedge Q$ are the fuzzy events with membership functions:

\vspace{-0.7cm}

$$\mu_{P\vee Q}=\mu_P+\mu_Q-\mu_P\star\mu_Q$$
$$\mu_{P\wedge Q}=\mu_P\star\mu_Q\hspace{2cm}$$
vii) An ordering set of states is formed by the probability measures $\{\sigma_{p}\}_{p\in\sP(\sH)}$ defined by $\sigma_p(P):=\mu_P(p)$.

\end{proposition}

\begin{proof}
The membership function $\mu:\sP(\sH)\rightarrow[0,1]$ of any fuzzy event of $\gF$ has the form $\mu(p)=\tr(Tp)$, with $T$ self-adjoint operator in $\sH$, by Theorem \ref{AS}. If $\mu_1,\mu_2\in\gF$ we have:
$$(\mu_1\star\mu_2)(p)=\tr(T_1p)\tr(T_2p)+\frac{i}{2}\{\mu_1, \mu_2\}_{PB} (p)+ \frac{1}{2} g_p(X_1(p),X_2(p)), $$
where $X_1$ and $X_2$ are the Hamiltonian vector fields of $\mu_1$ and $\mu_2$ respectively. By definitions (\ref{omega}) and (\ref{g}):
$$(\mu_1\star\mu_2)(p)=\tr(T_1p)\tr(T_2p)+\frac{1}{2}\tr([T_1,T_2]p)-\frac{1}{2}\tr(([T_1,p][T_2,p]+[T_2,p][T_1,p])p)$$
$$=\frac{1}{2}\tr([T_1,T_2]_+ p) +\frac{1}{2}\tr([T_1,T_2]p),\hspace{4.3cm}$$ 
where $[\,\,,\,\,]_+$ is the anti-commutator of operators. Since $[T_1,T_2]+[T_1,T_2]_+=2T_1T_2$ we obtain that $(\mu_1\star\mu_2)(p)=\tr(T_1T_2p)$ for all $p\in\sP(\sH)$. Thus if $\mu\in\gF$ then $\mu(p)=\tr(Tp)$ for a self-adjoint operator $T$ such that $T^2=T$, i.e. $T$ is an orthogonal projector in $\sH$. This fact establishes a bijective correspondence between $\gF$ and the class of orthogonal projectors in $\sH$ denoted by $\gP(\sH)$. In order to prove that $\gF$ is a quantum logic w.r.t. the fuzzy inclusion, we show that the map $h:\gP(\sH)\ni T\mapsto A_T\in\gF$, where the membership function of $A_T$ is $\mu_{A_T}(p)=\tr(Tp)$, is an order isomorphism. In this way $\gF$ can inherit the structures of $\gP(\sH)$ that is a bounded orthocomplemented, $\sigma$-complete, orthomodular lattice.
Let $\sS(\sH)$ be the unit sphere in $\sH$, if $\langle \psi|T\psi\rangle\leq \langle\psi|S\psi\rangle$ for any $\psi\in\sS(\sH)$ (i.e. $T\leq S$ in $\gP(\sH)$), that is $\tr(T\ket\psi\bra\psi)\leq \tr(S\ket\psi\bra\psi)$ for any $\psi\in\sS(\sH)$, then $\mu_{A_T}(p)\leq \mu_{A_S}(p)$ $\forall$ $p\in\sP(\sH)$ (i.e. $A_T\subseteq A_S$ as fuzzy sets). Conversely, if $\mu_{A_T}(p)\leq \mu_{A_S}(p)$ $\forall$ $p\in\sP(\sH)$ then $\tr(T\ket\psi\bra\psi)\leq \tr(S\ket\psi\bra\psi)$ for any $\psi\in\sS(\sH)$, i.e. $\langle \psi|T\psi\rangle\leq \langle\psi|S\psi\rangle$ for any $\psi\in\sS(\sH)$. Thus $\gF$ is a lattice w.r.t. the fuzzy inclusion that is isomorphic to $\gP(\sH)$. $h$ induces a structure of quantum logic on $\gF$ in the following way: The minimum in $\gP(\sH)$ is the null operator {\bfseries 0} and the maximum is the identity operator $\bI$ then $\gF$ is bounded where the minimum is the fuzzy set with zero membership function i.e. it is the empty set $\emptyset$ and the maximum is the fuzzy set with the constant membership function 1 i.e. the universe $\sP(\sH)$. The orthocomplement of $T$ in the lattice $\gP(\sH)$ is $\bI-T$, thus for $A\in\gF$ its orthocomplement, induced by $h$, is the fuzzy set $\neg A$ with membership function $\mu_{\neg A}(p)=\tr((\bI-h^{-1}(A))p)=1-\tr(h^{-1}(A)p)=1-\mu_A(p)$. The properties $iv), v), vi)$ are directly inherited by $\gP(\sH)$ via the isomorphism $h$. Claim vii) derives by the definition of partial order relation in $\gF$, in fact if $\sigma_p(P)\leq \sigma_p(Q)$ for any $p\in\sP(\sH)$ then $\mu_P(p)\leq \mu_Q(p)$ for any $p\in\sP(\sH)$ that is $P\leq Q$.
\end{proof}

Let us observe that any generalized probability measure on $\gF$ (Definition \ref{GPM}) presents the form of a state as defined in (\ref{fuzzy state}) for some Liouville density $\rho$ on $\sP(\sH)$ as a direct consequence of Theorem \ref{AS} and Gleason's theorem. In particular the states of the ordering set are the \emph{Dirac-like} measures introduced in (\ref{Fuzzy Dirac}).

In view of Proposition \ref{geologic} the class of fuzzy events in the complex projective space used to describe the quantum propositions within the geometric formulation presents a structure of quantum logic that is isomorphic to the standard B-vN quantum logic. In the next section we show that this structure is not merely induced by the lattice of orthogonal projectors in the underlying Hilbert space but derives from a general structure of quantum logic carried by collections of fuzzy sets where a non-commutative $\star$-product between membership functions is defined.

\section{Quantum logic by deformation of the product $t$-norm}\label{deform t-norm}

Proposition \ref{geologic} turns out to be a special case of a result proved in this section regardless the notion of projective Hilbert space. In fact the quantum structure of $\gF$  can be generalized in terms of fuzzy set-theoretic elements. Let us recall a known statement about functional aspects of quantum logics \cite{mcz, py2}.
\begin{theorem}\label{Mik}
Let $\mathcal E$ be a set of functions from a non-empty set $X$ to $[0,1]$ such that the following requirements are satisfied:
\\
i) $0\in\cE$;\\
ii) $f\in\cE\Rightarrow 1-f\in\cE$;
\\
iii) If $\{f_n\}_{n\in N}\subset \cE$ with $N$ finite or countable and $f_i+f_j\leq1$ for $i\not =j$ then $\sum_{n\in N} f_n \in\cE$.
\\
Then $\cE$ is a bounded orthocomplemented $\sigma$-orthocomplete orthomodular poset w.r.t. the partial order relation of real functions ($f\geq g$ if $f(x)\geq g(x)$ for any $x\in X$) where orthocomplementation is given by $\neg f=1-f$. 

Moreover for any $x\in X$ a generalized probability measure on $\cE$ is defined by $\sigma_x(f):=f(x)$ and $\{\sigma_x\}_{x\in X}$ is an ordering set on $L$.

\end{theorem}

The next result gives a procedure to obtain a quantum logic from a collection of fuzzy sets by means of the definition of \emph{deformed} fuzzy union and intersection where the pointwise product of the $t$-nom (\ref{product}) is replaced by a unital associative product $\star$ over membership functions. Then the deformed set operations will not satisfy all the requirements of union and intersection unless the membership functions $\star$-commute. Let us remark that we do not use the terms \emph{deformation} and \emph{$\star$-product} in the sense of deformation quantization as clarified by the statement of Theorem \ref{main}. However we will observe a possible connection between this result and some aspects of deformation quantization.

\begin{theorem}\label{main}
Let $\widetilde\gF$ be a class of fuzzy sets in the universe $\mathcal U$. Let $M$ be a set of functions $\mathcal U\rightarrow[0,1]$ such that $\mu_A\in M$ for any $A\in\widetilde\gF$ and $\star$ be a product\footnote{We mean that $\star$ is a binary operation over $M$ such that $(M,+,\star)$ is a ring where $+$ is the standard sum.} over $M$. For $A$ and $B$ in $\widetilde\gF$ let $A\Cup B$ and $A\Cap B$ be the fuzzy sets in $\mathcal U$ with membership functions:

\vspace{-0.5cm}

\beq\label{u}
\mu_{A\Cup B}:=\mu_A+\mu_B-\mu_A\star\mu_B ,
\eeq

\vspace{-0.8cm}

\beq\label{i}
\mu_{A\Cap B}:=\mu_A\star\mu_B . \hspace{2.1cm}
\eeq
Let $\gF$ be the subclass of elements of $\widetilde\gF$ such that their memebrship functions are idempotent w.r.t. $\star$. If the following facts hold:
\\
i) $\emptyset\in\gF$;\\
ii) $A\in\gF\Rightarrow\neg A\in\gF$;\\
iii) $\mu_A+\mu_B\leq 1$ if and only if $A\Cap B=\emptyset$;            \\
iv) If $\{A_i\}_{i\in\bN}\subset \gF$ satisfies $A_i\Cap A_j=\emptyset$ for $i\not =j$  then $\Cup_i A_i\in\gF$;\\
then:
\\
1) $\gF$ is a quantum logic w.r.t the fuzzy order relation where the orthocomplementation is given by the fuzzy complement;\\
2) $A\perp B$ if and only if $\mu_A\star\mu_B=0$. 
\\
3) If $\gF_0\subset \gF$ is a lattice such that $A\vee B=A\Cup B$ and $A\wedge B=A\Cap B$ for all $A, B\in\gF_0$ then all the fuzzy sets of $\gF_0$ are pairwise compatible and their membership functions commute w.r.t. $\star$.
\end{theorem}

\begin{proof}
Let $\gF$ be a collection of fuzzy sets that satisfies hypotheses i)-iv) and $\cE$ be the corresponding set of membership functions. $0\in\cE$ by hypothesis i). If $\mu\in\cE$ then $1-\mu\in\cE$ by hypothesis ii). Let us consider $\{\mu_i\}_{i\in\bN}\subseteq \cE$ such that $\mu_i+\mu_j\leq 1$ for $i\not =j$ then, by hypothesis iii), the corresponding family of fuzzy sets $\{A_i\}_{i\in\bN}$ satisfies $A_i\Cap A_j=\emptyset $ for $i\not =j$. By hypothesis iv) we have that $\Cup_i A_i\in \gF$, so  $\sum_i \mu_i\in\cE$. In view of Theorem \ref{Mik} we have that $\cE$ is a quantum logic w.r.t. the partial order relation of real functions and the complement $\neg\mu=1-\mu$. Therefore $\gF$ turns out to be a quantum logic w.r.t. the fuzzy inclusion and the fuzzy complement as orthocomplementation.
\\
Statement 2) directly derives form hypothesis iii): If $A\perp B$ in $\gF$, i.e. $A\leq \neg B$, then $\mu_A\leq 1-\mu_B$. By hypothesis iii): since $\mu_A+\mu_B\leq 1$ we have $\mu_A\star\mu_B=0$. Conversely, if $\mu_A\star\mu_B=0$ then $\mu_A\leq 1-\mu_B$. 
\\
Assuming that $\gF_0\subset \gF$ is a lattice  where $A\vee B=A\Cup B$ for all $A, B\in\gF_0$, since $A\vee B=B\vee A$ we have that the membership functions of the elements in $\gF_0$ obviously commute w.r.t. $\star$. Moreover it is straightforward checking that $(\gF_0, \Cup, \Cap, \neg, \emptyset, \mathcal U)$ is a Boolean algebra by idempotence, then any pair of elements in $\gF_0$ commute.
\end{proof}

\vspace{0.5cm}

In the previous section, Proposition \ref{geologic} establishes the structure of quantum logic of a class of fuzzy events (fuzzy sets with Borel measurable membership functions) in a Hilbert projective space equipped with a particular non-commutative product between membership functions. In that case $\gF$ turns out to be the geometric Hamiltonian version of the standard Birchoff-von Neumann quantum logic. Theorem \ref{main} is a more general result that establishes the structure of a quantum logic of a class of fuzzy sets equipped with a general non-commutative product between membership functions where the compatibility is related to the commutativity w.r.t.\!\! the considered product. In particular (\ref{u}) and (\ref{i}) are modifications of fuzzy union and intersection given by the product t-norm (\ref{product}) where the pointwise product of membership functions is replaced by a (in general non-commutative) product $\star$ over the set $M$. In Proposition \ref{geologic}, $M$ corresponds to $\cC^\infty_\bR(\sP(\sH))$.

Let us conclude this section observing a possible relationship of Theorem \ref{main} with deformation quantization. \emph{Deformation quantization} can be summarized by the following paradigma: Given the phase space $\cM$ of a classical system, the observable algebra $\cC^\infty(\cM)$ is deformed into a non-commutative one in order to obtain a quantum observable algebra. For example, $\cC^\infty(\bR^2)$ is deformed by the Moyal product. More generally a Poisson algebra $A$ is deformed by means of a formal power series in $\hbar$:
\beq
f\star g:=f\cdot g+\sum_{n=1}^\infty \hbar^n P_n(a,b),\quad f,g\in A
\eeq
where $P_n$ are $n$-bilinear forms such that $\star$ is associative \cite{ger}. If we have a classical system with phase space $\cM$ then we describe the elementary propositions about it as Borel sets in $\cM$. A rough prescription of quantization can be formulated in these terms: Quantum propositions are described by fuzzy sets with membership functions in $\cC^\infty(\cM)$ (e.g. mollifications $\mu_P$ of the indicator functions $\chi_P$ of the Borel sets $P$) equipped with a non-commutative $\star$-product such that the hypotheses of Theorem \ref{main} are satisfied. In this way we obtain a structure of a quantum logic where the notion of quantum state can be provided as a generalized probability measure and a non-commutative product is defined on the observable algebra. In order to recover a notion of classical limit one should require that $f\star g\rightarrow f\cdot g$ for all $f,g\in\cC^\infty(\cM)$ and $\mu_P\rightarrow\chi_P$ for all $P\in\cB(\cM)$ as $\hbar\rightarrow 0$. However the development of these ideas is beyond the contents and scope of this work.

\section{Conclusions}

In this paper we have constructed a concrete quantum logic, a $\sigma$-orthocomplete orthomodular poset allowing an ordering set of generalized probability measures, within the geometric formulation of quantum mechanics. More precisely we have described propositions about a quantum system as fuzzy sets in the quantum phase space by analogy with Hamiltonian classical mechanics where propositions are represented by Borel sets in the classical phase space. Then we have proved that the class $\gF$ of fuzzy sets in a projective Hilbert space $\sP(\sH)$, representing quantum propositions, presents a structure of a quantum logic. In particular $\gF$ turns out to be a bounded $\sigma$-complete orthomodular lattice that is isomorphic to the non-distributive lattice of orthogonal projectors on $\sH$. However $\gF$ is not merely a rephrasing of the standard Birkhoff-von Neumann quantum logic but it is a special case of a general quantum logic structure carried by classes of fuzzy sets endowed with a non-commutative product over membership functions. The statement of theorem \ref{main} can be interpreted as a prescription to deform the product $t$-norm of fuzzy intersection by means of a (in general non-commutative) product $\star$ in order to obtain a quantum logic from a class of fuzzy sets such that the compatibility in the poset is related to the commutation w.r.t. $\star$. Indeed quantum logics of fuzzy sets, the so-called \emph{fuzzy quantum logics}, are an old idea introduced in \cite{py1} and a complete characterization of them was provided by Pykacz in \cite{py2} in terms of collections of fuzzy sets satisfying a list of requirements w.r.t.\! Lukasiewicz operations. In this regard, following the geometric classical-like approach to quantum mechanics, we have introduced a specific kind of fuzzy quantum logics (i.e.\! those satisfying the hypotheses of Theorem \ref{main}) whose \emph{quantumness} is characterized by a non-commutative product, in this spirit we suggest a general connection with deformation quantization. Therefore a possible direction of investigation could be devoted to the construction of a quantum logic of fuzzy sets via a deformed $t$-norm as a quantization prescription.\!\! An interesting issue may be the relationship between the deformations of the product $t$-norm for operations on fuzzy sets in a classical phase space and the $\star$-products of deformation quantization.

\section*{Acknowledgements}
This work is supported by a grant from Q@TN.


\begin{thebibliography}{References}

\bibitem[ATV83]{als} C. Alsina, E. Trillas, L. Valverde {\em On some logical connectives for fuzzy set theory}, J. Math. Anal. Appl. 93, 15-26 (1983)
\bibitem[AS95]{AS}  A. Ashtekar and T. A. Schilling, {\em Geometry of quantum mechanics}, AIP Conference Proceedings, {\bf 342}, 471-478 (1995)
\bibitem[BSS04]{BSS} A. Benvegn\`u, N. Sansonetto and M. Spera \emph{Remarks on geometric quantum mechanics}, Journal of Geometry and Physics 51 229-243 (2004)
\bibitem[BvN36]{BvN} G. Birkhoff, J. von Neumann {\em The logic of quantum mechanics}, Ann. Math. 37, 823-843 (1936)
\bibitem[BH01]{BH} D. C. Brody, L. P. Hughston. {\em Geometric quantum mechanics},  J. Geom. Phys. {\bf 38} 19-53  (2001)
\bibitem[BW85]{convex} L. J. Bunce, J. D. M. Wright  \emph{Quantum logics and convex geometry}, Commun. Math. Phys. 101, 87-96 (1985)
\bibitem[Ca17]{camosso} S. Camosso {\em Quantum logic and geometric quantization}, J. Quantum Inf. Sci., vol. 7, 35-42 (2017)
\bibitem[dCK16]{CK} N. da Costa, J. Kouneiher {\em Superlogic manifolds and geometric approach to quantum logic}, Int. J. Geom. Methods M. Vol. 13, No. 01, 1550129 (2016)
\bibitem[EGL07]{HQL} K. Engesser, D. M. Gabbay, D. Lehmann (eds.) {\em Handbook of Quantum Logic and Quantum Structures}, Elsevier, Amsterdam (2007)
\bibitem[FR74]{FR} D. J. Foulis, C. H. Randall, {\em Empirical logic and quantum mechanics}, Synthese 29, 81-111 (1974)
\bibitem[Ge63]{ger} M. Gerstenhaber {\em The cohomology structure of an associative ring}, Ann. Math. 78(2):267-288 (1963)
\bibitem[Gle57]{Gleason} A. M. Gleason. \emph{Measures on the closed subspaces of a Hilbert space}, Journal of Mathematics and Mechanics, Vol.6, No.6, 885-893 (1957).
\bibitem[Kib79]{Kibble} T. W. B. Kibble {\em Geometrization of Quantum Mechanics}, Commun. Math. Phys. { 65}, 189-201 (1979)
\bibitem[Mac57]{Mac} G. W. Mackey {\em Quantum mechanics and Hilbert space}, Amer. Math. Monthly 64:2, 45-57 (1957)
\bibitem[Mcz74]{mcz} M. J. Maczy\'nski {\em Functional properties of quantum logics}, Int. J. Theor. Phys., 11, 149-156 (1974)
\bibitem[Mo17]{V} V. Moretti {\em Spectral Theory and Quantum Mechanics}, (2nd ed.) Springer International Publishing (2017)
\bibitem [MP16]{DV2}  V. Moretti and D. Pastorello \emph{Frame functions in finite-dimensional quantum mechanics and its hamiltonian formulation on complex projective spaces}, Int. J. Geom. Methods M. Vol. 13, No. 02, 1650013 (2016)
\bibitem[Pi64]{piron} C. Piron {\em Axiomatique quantique}, Helvetica Physica Acta 37, 439-468 (1964)
\bibitem[Py87]{py1} J. Pykacz {\em Quantum logics as families of fuzzy subsets of the set of physical states}, Preprints of the Second International Fuzzy Systems Association Congress, Tokyo, July 20-25, vol2, pp. 437-440 (1987)
\bibitem[Py94]{py2} J. Pykacz {\em Fuzzy quantum logics and infinite-valued Lukasiewicz logic}, Int. J. Theor. Phys., 33, 1403-1416 (1994)
\bibitem[Py07]{py3} J. Pykacz {\em Quantum structures and fuzzy set theory} in Handbook of Quantum Logic and Quantum Structures, Elsevier Amsterdam (2007)
\bibitem[Za65]{Zadeh} L. A. Zadeh {\em Fuzzy sets}, Information and Control 8 (3) 338-353 (1965)
\bibitem[Za68]{Zadeh2} L. A. Zadeh {\em Probability measures of fuzzy events}, J. Math. Anal. Appl. 23, 421-427 (1968)
\end{thebibliography}
\end{document}